\documentclass[11pt]{article}

\usepackage[letterpaper,margin=1in]{geometry}
\usepackage[T1]{fontenc}
\usepackage{lmodern}
\usepackage{microtype}
\usepackage{amsmath,amssymb,amsthm}
\usepackage{xcolor}

\definecolor{arxivblue}{RGB}{0,80,160}
\usepackage[
    colorlinks=true,
    allcolors=arxivblue
]{hyperref}
\usepackage{orcidlink}
\hypersetup{
    pdftitle={Parallel repetition in the two-player quantum cloning game},
    pdfauthor={Eli Coe Naig and Stephen A. Fenner},
    pdfsubject={Quantum information and parallel repetition}
}

\theoremstyle{definition}
\newtheorem{definition}{Definition}[section]
\theoremstyle{plain}
\newtheorem{theorem}[definition]{Theorem}
\newtheorem{proposition}[definition]{Proposition}
\newtheorem{lemma}[definition]{Lemma}
\newtheorem{corollary}[definition]{Corollary}
\theoremstyle{remark}
\newtheorem{remark}[definition]{Remark}

\newcommand{\id}{\mathbb I}
\newcommand{\Tr}{\operatorname{Tr}}
\newcommand{\ket}[1]{\lvert #1\rangle}
\newcommand{\bra}[1]{\langle #1\rvert}
\newcommand{\ketbra}[2]{\lvert #1\rangle\!\langle #2\rvert}
\newcommand{\norm}[1]{\lVert #1\rVert}
\newcommand{\QCG}{\mathrm{QCG}}
\newcommand{\two}{{\{0,1\}}}

\newcommand{\cK}{\mathcal{K}}
\newcommand{\SWAP}{\mathrm{SW}}
\newcommand{\Rreg}{{\mathsf R}}
\newcommand{\Areg}{{\mathsf A}}
\newcommand{\Breg}{{\mathsf B}}

\title{Parallel repetition in the two-player quantum cloning game}

\author{
  Eli Coe Naig\,\orcidlink{0009-0003-9324-2382}$^{\,1,2,3}$ and
  Stephen A.\ Fenner\,\orcidlink{0000-0002-6028-9323}$^{\,2}$\thanks{Corresponding author: \url{fenner.sa@gmail.com}.  This work is partially supported by the U.S. National Science Foundation's Research Experiences for Undergraduates award \#2447817.}\\[0.4em]
  \small $^{1}$Tulane University, New Orleans, Louisiana, USA\\
  \small $^{2}$Department of Computer Science and Engineering,\\
  \small University of South Carolina, Columbia, South Carolina, USA\\
  \small $^{3}$QSPARC Labs, USA
}

\date{August 1, 2026}

\begin{document}
\maketitle

\begin{abstract}
    We study parallel repetition in the two-player quantum cloning game, a
    monogamy-of-entanglement game motivated by quantum position verification.
    Colisson Palais, Escol\`a-Farr\`as, and Speelman bounded the value of
    $n$ copies between $(3/4)^n$ and $\cos^{2n}(\pi/8)$. For two copies, we
    prove that neither bound is tight.
    An explicit challenge-dependent strategy achieves value
    $(5+\sqrt{17})/16>9/16$, so strong parallel repetition fails for the
    unrestricted game. A block Gram
    matrix argument gives the upper bound
    $(11+\sqrt{65})/32<\cos^4(\pi/8)$ and strictly improves the previous
    parallel-repetition upper bound for every $n$. For every $n$, challenge-independent
    strategies have optimal value $(3/4)^n$.
\end{abstract}

\section{Introduction}\label{sec:intro}

In the two-player quantum cloning game $\QCG_2$~\cite{CES25}, a referee and
two noncommunicating players share a state.  The referee samples a uniform bit $b\in\two$,
announces $b$ to both players, and tests whether the player designated by $b$ shares a
Bell pair with the referee.  The optimum probability of this event over all possible strategies of the two players is denoted $\omega^*(\QCG_2)$.  The value for a single repetition is
$\omega^*(\QCG_2)=3/4$~\cite[Theorem~6]{CES25}.
The game is motivated by quantum position verification, part of the
position-based cryptography program introduced by Chandran et
al.~\cite{Chandran09}.  In the quantum setting, Buhrman et al.~\cite{Buhrman14}
showed that no protocol is secure against adversaries sharing arbitrarily large
entanglement, and constructed protocols secure against adversaries with none.

In $n$ parallel repetitions, the referee announces a string
$x\in\two^n$, and each player may respond to the full string.  Colisson
Palais, Escol\`a-Farr\`as, and Speelman~\cite[Theorem~8]{CES25} proved
\begin{equation}\label{eq:previous}
    \left(\frac34\right)^n
    \le \omega^*(\QCG_2^{\times n})
    \le \left(\frac12+\frac1{2\sqrt2}\right)^n
    =\cos^{2n}(\pi/8),
\end{equation}
and left the tightness of both bounds open.
Their upper bound adapts the projection-overlap method of Tomamichel et
al.~\cite{TFKW13}.

We show that neither endpoint is tight at $n=2$:
\begin{equation}\label{eq:main-interval}
    \left(\frac34\right)^2
    < \frac{5+\sqrt{17}}{16}
    \le \omega^*(\QCG_2^{\times2})
    \le \frac{11+\sqrt{65}}{32}
    < \cos^4(\pi/8).
\end{equation}
Strong parallel repetition would mean
$\omega^*(\QCG_2^{\times n})=\omega^*(\QCG_2)^n$ for all $n$.  Because
$\omega^*(\QCG_2)=3/4$, the leftmost inequality in \eqref{eq:main-interval}
shows that strong parallel repetition fails at $n=2$.  The rightmost
inequality strictly improves the upper bound of \eqref{eq:previous} at
$n=2$, and Proposition~\ref{prop:general} does so for every $n$.

A block Gram matrix indexed by the challenge strings expresses the game
value as an operator norm.  For a pair of challenges $x,y$, the norm of the
corresponding off-diagonal block depends on the \emph{direction} of each
differing bit: whether that coordinate switches from testing Alice under $x$
to testing Bob under $y$, or the reverse (Section~\ref{sec:caps}).  The resulting upper bounds on the block norms,
which we call \emph{caps}, assemble into the cap matrix of
Definition~\ref{def:cap}.  The proof of \cite[Theorem~8]{CES25} obtains an oriented overlap factor
$2^{-t_A}$ and, after exchanging the two challenges so that $t_A\ge t/2$, retains
only the Hamming-distance bound $2^{-t/2}$ in its aggregation (notation of
Section~\ref{sec:caps}).  Our contribution
is to retain the full directional information simultaneously in the block Gram
matrix, yielding a strictly smaller upper bound for every $n$;
to give the explicit challenge-dependent two-copy strategy; and to show
(Proposition~\ref{prop:challenge-independent}) that for every $n$ the optimal
challenge-independent value is exactly
 $(3/4)^n$, the $n$th power of the one-copy
value~\cite[Theorem~6]{CES25}, so the counterexample necessarily uses
challenge-dependent responses.

Section~\ref{sec:model} casts the game value as the operator norm of a
block Gram matrix indexed by the challenges.  Section~\ref{sec:caps}
bounds the off-diagonal blocks and assembles the cap matrix,
Section~\ref{sec:attack} presents the two-copy counterexample and the
exact challenge-independent value for every $n$, and Section~\ref{sec:upper} derives
the two-copy and general upper bounds.

\section{A block Gram matrix for the cloning game}\label{sec:model}

We recall the $n$-fold parallel repetition of the cloning game
of~\cite{CES25}.  Let $R_1,\ldots,R_n$ be the referee's qubits.  Alice and Bob hold output
qubits $A_1,\ldots,A_n$ and $B_1,\ldots,B_n$, together with arbitrary finite-dimensional
private registers $E_A$ and $E_B$.  Write
\[
    \Rreg=R_1\cdots R_n,\qquad
    \Areg=A_1\cdots A_n,\qquad
    \Breg=B_1\cdots B_n
\]
for the collective registers.  For $b\in\{0,1\}$, set
\[
    Q_{b,j}=\begin{cases}A_j,& \mbox{if $b=0$,}\\ B_j,& \mbox{if $b=1$.}\end{cases}
    \qquad
    \ket{\Phi^+}=\frac{\ket{00}+\ket{11}}{\sqrt2}.
\]
For a challenge $x\in\{0,1\}^n$, define
\[
    \Pi_x=\bigotimes_{j=1}^n
    \ketbra{\Phi^+}{\Phi^+}_{R_jQ_{x_j,j}}\otimes\id_{\mathrm{rest}},
\]
where the identity acts on the remaining registers: the untested output
qubit $Q_{1-x_j,j}$ at each coordinate $j$ and the private registers $E_A$
and $E_B$, so that $\Pi_x$ is a projector on
$\Rreg\otimes\Areg\otimes\Breg\otimes E_A\otimes E_B$.
A strategy consists of an initial state $\rho$ over the space of all registers and, for each challenge $x$, a
pair of \emph{local response} channels.  Each channel admits a Stinespring
isometry (for background, see~\cite{Watrous18}); appending a fixed pure state on an enlarged
private register, extending the isometry to a unitary, and retaining the
environment inside the private register lets us represent the responses by
challenge-dependent local unitaries
$U^x_{\Areg E_A}$ and $V^x_{\Breg E_B}$.  Since the initial state is
unrestricted, the fixed enlargement does not change the supremum, and
the Bell tests act as the identity on the private registers.  Write
\[
    S_x=\id_{\Rreg}\otimes U^x_{\Areg E_A}\otimes V^x_{\Breg E_B}.
\]
The game value is
\begin{equation}\label{eq:value}
    \omega^*(\QCG_2^{\times n})
    =\sup_{\rho,\{U^x,V^x\}_x}\frac1{2^n}\sum_x
    \Tr\!\left[
        \Pi_xS_x\rho S_x^\dagger
        \right].
\end{equation}
No constraint is placed on the shared initial state or on the referee's
marginal; we therefore call the game \emph{unrestricted}, in contrast to
the routing model of \cite[Section~5]{CES25}.

Write $C_x=\bra{\Phi_x}\otimes\id_{\mathrm{rest}}$, where
$\ket{\Phi_x}=\bigotimes_j\ket{\Phi^+}_{R_jQ_{x_j,j}}$, and define the
success map
\[
    A_x=C_xS_x
    \colon\mathcal H\to\mathcal K_x,
\]
where $\mathcal H$ is the full pre-measurement space and $\mathcal K_x$ collects
the subsystems left after the challenge-$x$ Bell contractions together with all
private registers.  Then $C_x^\dagger C_x=\Pi_x$ and
$C_xC_x^\dagger=\id_{\mathcal K_x}$, so each $C_x$ is a coisometry and
$A_x^\dagger A_x=S_x^\dagger\Pi_xS_x$.

By cyclicity of the trace,
\[
    \Tr\!\left[\Pi_xS_x\rho S_x^\dagger\right]
    =\Tr\!\left[A_x^\dagger A_x\,\rho\right],
\]
and since $\sum_xA_x^\dagger A_x$ is positive semidefinite, the supremum over
density operators $\rho$ of the average in \eqref{eq:value} equals
$\tfrac1{2^n}\norm{\sum_xA_x^\dagger A_x}$.  Stacking the maps into
\[
    \mathcal A\colon\mathcal H\to\bigoplus_x\mathcal K_x,
    \qquad
    \mathcal A\psi=(A_x\psi)_x,
\]
gives
\[
    \mathcal A^\dagger\mathcal A=\sum_xA_x^\dagger A_x,
    \qquad
    \mathcal A\mathcal A^\dagger=[M_{x,y}]_{x,y},
\]
where $M_{x,y}=A_xA_y^\dagger$ are the entries of the block Gram matrix.
The two operators share the same nonzero spectrum, so
$\norm{\sum_xA_x^\dagger A_x}=\norm{[M_{x,y}]_{x,y}}$.  Taking the supremum
over the response unitaries gives
\begin{equation}\label{eq:gram}
    \omega^*(\QCG_2^{\times n})
    =\frac1{2^n}\sup_{\{U^x,V^x\}_x}
    \norm{[M_{x,y}]_{x,y}}.
\end{equation}
If
\begin{equation}\label{eq:Wxy}
    W^{xy}=S_xS_y^\dagger=\id_{\Rreg}
    \otimes U^x_{\Areg E_A}\bigl(U^y_{\Areg E_A}\bigr)^\dagger
    \otimes V^x_{\Breg E_B}\bigl(V^y_{\Breg E_B}\bigr)^\dagger,
\end{equation}
then $M_{x,y}=C_xW^{xy}C_y^\dagger$, and
$M_{x,x}=C_xC_x^\dagger=\id_{\cK_x}$.

\begin{lemma}[Norm-matrix bound]\label{lem:norm-matrix}
    Let $B = [B_{x,y}]$ be a square matrix with equal size square blocks $B_{x,y}$ for $x,y$ running over the same domain,

    and let $N$ be an entrywise nonnegative matrix satisfying
    $\norm{B_{x,y}}\le N_{x,y}$ for all $x,y$, where $\norm{\cdot}$ is the operator norm. Then
    $\norm B\le\norm N$.

\end{lemma}

\begin{proof}
    For unit block vectors $u=(u_x)_x$ and $v=(v_y)_y$, define
    $a_x=\norm{u_x}$ and $b_y=\norm{v_y}$. Then
    $\norm a=\norm b=1$. By the triangle inequality,
    Cauchy--Schwarz, and the definition of the operator norm,
    \[
        |\langle u,Bv\rangle|
        =
        \left|
        \sum_{x,y}\langle u_x,B_{x,y}v_y\rangle
        \right|
        \le
        \sum_{x,y}
        \norm{u_x}\norm{B_{x,y}v_y}
        \le
        \sum_{x,y}
        \norm{u_x}\norm{B_{x,y}}\norm{v_y}
        \le
        a^{\mathsf T}Nb
        \le
        \norm N.
    \]
    Taking the supremum over unit $u$ and $v$ proves the claim.
\end{proof}

\section{Challenge-pair bounds}\label{sec:caps}

For $x,y\in\{0,1\}^n$, define
\[
    T_A=\{j:x_j=0,\ y_j=1\},\qquad
    T_B=\{j:x_j=1,\ y_j=0\},
\]
and let $t_A=|T_A|$, $t_B=|T_B|$, and $t=t_A+t_B$, written $t_A(x,y)$ and
$t_B(x,y)$ when the ordered pair needs emphasis.  For $x\ne y$, call the pair
a \emph{chain} if $t_A=0$ or $t_B=0$, and a \emph{loop} otherwise, that is,
if $t_A>0$ and $t_B>0$.  Thus all
changes in a chain point in one direction, whereas a loop contains changes in
both directions.

Set
\[
    P_x=S_x^\dagger\Pi_xS_x=A_x^\dagger A_x.
\]
These are the conjugated winning projectors denoted $M^x$ in the proof of
\cite[Theorem~8]{CES25}; despite the letter, their $M^x$ acts on the full
strategy space and is distinct from our block Gram entries $M_{x,y}$.
Each $A_x=C_xS_x$ is a coisometry, because $C_x$
is one (Section~\ref{sec:model}) and $S_x$ is unitary; hence
\[
    P_xP_y=A_x^\dagger M_{x,y}A_y,
    \qquad
    M_{x,y}=A_x\bigl(P_xP_y\bigr)A_y^\dagger.
\]
Since $\norm{A_x}=\norm{A_y}=1$, the two factorizations give
\[
    \norm{M_{x,y}}=\norm{P_xP_y}.
\]
We now follow the projector-relaxation argument of
\cite[Theorem~8]{CES25}, retaining its directional dependence on $t_A$.

\medskip\noindent\textit{Single-coordinate calculation.}
On three qubits $R,A,B$, define
\begin{align}
    \Gamma^A&=\ketbra{\Phi^+}{\Phi^+}_{RA}\otimes\id_B,
    &
    \Gamma^B&=\ketbra{\Phi^+}{\Phi^+}_{RB}\otimes\id_A, \label{eq:GammaAB}\\
    J_A&=\ket{\Phi^+}_{RA}\otimes\id_B,
    &
    J_B&=\ket{\Phi^+}_{RB}\otimes\id_A, \label{eq:JAB}
\end{align}

so that $\Gamma^A=J_AJ_A^\dagger$ and $\Gamma^B=J_BJ_B^\dagger$.  The
isometry $J_B$ takes $A$ into $R\otimes A\otimes B$, and $J_A^\dagger$
takes $R\otimes A\otimes B$ onto $B$.  For $\ket\psi\in A$,
\[
    J_A^\dagger J_B\ket\psi
    =\frac12\sum_{a,b\in\{0,1\}}
    \langle a\vert b\rangle_R\,
    \langle a\vert\psi\rangle_A\,\ket b_B
    =\frac12\sum_{a\in\{0,1\}}
    \langle a\vert\psi\rangle_A\,\ket a_B,
\]
because $\langle a\vert b\rangle_R=\delta_{a,b}$, the Kronecker delta.  Hence
\begin{equation}\label{eq:Lambda}
    J_A^\dagger J_B=\frac12\,\Lambda,
    \qquad
    \Lambda=\sum_{a\in\{0,1\}}\ket a_B\bra a_A,
\end{equation}
where $\Lambda\colon A\to B$ is the identification of their
computational bases.  Left multiplication by the isometry $J_A$ and right
multiplication by $J_B^\dagger$ preserve the operator norm.  This can be seen as follows: for any $X$ such that $J_AXJ_B^\dagger$ is well-defined, let $\ket\psi$ and $\ket\phi$ be unit vectors such that $\norm{X} = \norm{X\ket\psi}$ and $\norm{J_AXJ_B^\dagger} = \norm{J_AXJ_B^\dagger\ket\phi}$.  Then $J_B\ket\psi$ is also a unit vector, and we have
\begin{align*}
    \norm{X} &= \norm{X\ket\psi} = \norm{XJ_B^\dagger J_B\ket\psi} = \norm{J_AXJ_B^\dagger (J_B\ket\psi)} \le \norm{J_AXJ_B^\dagger} \\
    &= \norm{J_AXJ_B^\dagger\ket\phi} = \norm{XJ_B^\dagger\ket\phi}
    \le \norm{X}\norm{J_B^\dagger\ket\phi} = \norm X\norm{J_BJ_B^\dagger\ket\phi} \le \norm X,
\end{align*}
the last inequality because $J_BJ_B^\dagger$ is a projector.  Thus all quantities are equal.

Therefore,
\begin{equation}\label{eq:GammaAB-norm}
    \norm{\Gamma^A\Gamma^B}
    =\norm{J_A\bigl(J_A^\dagger J_B\bigr)J_B^\dagger}
    =\norm{J_A^\dagger J_B}
    =\frac12.
\end{equation}

\begin{lemma}[Directional overlap bound]\label{lem:chain}
    For $x\ne y$,
    \[
        \norm{M_{x,y}}\le2^{-t_A}
        \qquad\text{and}\qquad
        \norm{M_{x,y}}\le2^{-t_B};
    \]
    hence $\norm{M_{x,y}}\le2^{-\max(t_A,t_B)}$.  In particular, if $(x,y)$
    is a chain, then $\norm{M_{x,y}}\le2^{-t}$.
\end{lemma}

\begin{proof}
    We prove $\norm{M_{x,y}}\le2^{-t_A}$ for an arbitrary ordered pair
    $(x,y)$ with $x\ne y$; since $M_{y,x}=A_yA_x^\dagger=M_{x,y}^\dagger$,
    $t_A(y,x)=t_B(x,y)$, and taking adjoints preserves the operator norm,
    applying this bound to $(y,x)$ then gives
    $\norm{M_{x,y}}\le2^{-t_B}$.  If $t_A=0$ the bound is immediate, because
    $\norm{M_{x,y}}=\norm{P_xP_y}\le1$; so assume $t_A\ge1$.

    For some $j\in T_A$ we have $x_j=0$ and $y_j=1$: challenge $x$ tests the pair
    $R_jA_j$, while challenge $y$ tests $R_jB_j$.  Keep only these tests and
    conjugate by the response unitaries:
    \[
        \widehat\Pi_x^A=\bigotimes_{j\in T_A}
        \ketbra{\Phi^+}{\Phi^+}_{R_jA_j}\otimes\id_{\mathrm{rest}},
        \qquad
        \widehat\Pi_y^B=\bigotimes_{j\in T_A}
        \ketbra{\Phi^+}{\Phi^+}_{R_jB_j}\otimes\id_{\mathrm{rest}},
    \]
    \[
        P_x^A=S_x^\dagger\widehat\Pi_x^AS_x,
        \qquad
        P_y^B=S_y^\dagger\widehat\Pi_y^BS_y,
    \]
    where each identity acts on every register not shown, including the
    referee qubits $R_j$ with $j\notin T_A$.  Because $\widehat\Pi_x^A$
    consists of a subset of the commuting tensor factors of $\Pi_x$, we have
    $\Pi_x=\Pi_x\widehat\Pi_x^A$, and conjugating by $S_x$ gives
    $P_x=P_xP_x^A$; likewise $P_y=P_y^BP_y$.  Hence
    $P_xP_y=P_x\bigl(P_x^AP_y^B\bigr)P_y$, and since
    $\norm{P_x},\norm{P_y}\le1$, multiplying by $P_x$ and $P_y$ cannot
    increase the operator norm:
    \[
        \norm{P_xP_y}\le\norm{P_x^AP_y^B}.
    \]

    The projector $\widehat\Pi_x^A$ acts as the identity on Bob's registers
    $\Breg E_B$, so conjugating it by $V^x_{\Breg E_B}$ or
    $V^y_{\mathsf BE_B}$ has no effect; similarly, $\widehat\Pi_y^B$ acts as
    the identity on Alice's registers $\Areg E_A$, so conjugating it by
    $U^x_{\Areg E_A}$ or $U^y_{\Areg E_A}$ has no effect.  Both
    relaxed projectors are therefore conjugated by the single unitary
    \[
        S_{x,y}=\id_{\Rreg}\otimes U^x_{\Areg E_A}\otimes V^y_{\Breg E_B}.
    \]
    Then
    \[
        P_x^A=S_{x,y}^\dagger\widehat\Pi_x^AS_{x,y},
        \qquad
        P_y^B=S_{x,y}^\dagger\widehat\Pi_y^BS_{x,y},
        \qquad
        P_x^AP_y^B
        =S_{x,y}^\dagger
        \bigl(\widehat\Pi_x^A\widehat\Pi_y^B\bigr)S_{x,y},
    \]
    since the inner $S_{x,y}S_{x,y}^\dagger$ cancels in the product.

    It remains to compute $\norm{\widehat\Pi_x^A\widehat\Pi_y^B}$.  For
    $j\in T_A$, let $\Gamma_j^A$ and $\Gamma_j^B$ denote the projectors of
    the single-coordinate calculation on the three qubits $R_jA_jB_j$ (see (\ref{eq:GammaAB})).
    Grouping the registers as $\bigotimes_{j\in T_A}(R_jA_jB_j)$ tensored
    with everything else and multiplying tensor factor by tensor factor,
    \[
        \widehat\Pi_x^A\widehat\Pi_y^B
        =\bigotimes_{j\in T_A}\bigl(\Gamma_j^A\Gamma_j^B\bigr)
        \otimes\id_{\mathrm{rest}}.
    \]
    The operator norm is unitarily invariant and multiplicative under
    tensor products, and each factor has norm $\tfrac12$ by Eq.~(\ref{eq:GammaAB-norm}) in the
    single-coordinate calculation, so
    \[
        \norm{M_{x,y}}
        =\norm{P_xP_y}
        \le\norm{P_x^AP_y^B}
        =\norm{\widehat\Pi_x^A\widehat\Pi_y^B}
        =\prod_{j\in T_A}\norm{\Gamma_j^A\Gamma_j^B}
        =2^{-t_A}.
    \]
    Finally, for a chain one of $t_A,t_B$ is zero, so
    $\max(t_A,t_B)=t_A+t_B=t$.
\end{proof}

Thus the directional estimate from the proof of \cite[Theorem~8]{CES25} can
be retained entrywise as $2^{-\max(t_A,t_B)}$ in the cap matrix below.

\medskip\noindent\textbf{Example: the pair $(00,01)$.}
Here $T_A=\{2\}$ and $T_B=\varnothing$: both challenges test $R_1A_1$,
while $x=00$ tests $R_2A_2$ and $y=01$ tests $R_2B_2$.  Removing the common
test on $R_1A_1$ leaves
\[
    \widehat\Pi_{00}^A=\ketbra{\Phi^+}{\Phi^+}_{R_2A_2}\otimes\id_{\mathrm{rest}},
    \qquad
    \widehat\Pi_{01}^B=\ketbra{\Phi^+}{\Phi^+}_{R_2B_2}\otimes\id_{\mathrm{rest}}.
\]
Their product is $\Gamma_2^A\Gamma_2^B\otimes\id_{\mathrm{rest}}$, of norm
$\tfrac12$ by (\ref{eq:GammaAB-norm}), so
$\norm{M_{00,01}}\le\tfrac12$ for every choice of response unitaries.

\begin{proposition}[Tightness of the loop bound]\label{prop:loop}
    There exist response unitaries for which the two-copy loop pair $(01,10)$
    satisfies $\norm{M_{01,10}}=1/2$, matching the bound
    $2^{-\max(t_A,t_B)}=2^{-1}$ of Lemma~\ref{lem:chain}.
\end{proposition}

\begin{proof}
    Take the private registers to be one-dimensional.  Set $x=01$ and $y=10$,
    and let the relative response
    unitaries swap the two output qubits on each side, exchanging
    $A_1\leftrightarrow A_2$ and $B_1\leftrightarrow B_2$ and acting as the
    identity on the private registers; this is realized, for example, by
    taking $U^{01}_{\Areg}=V^{10}_{\Breg}=\mathrm{SW}$ and every
    other response unitary the identity.  Plugging these into Eq.~(\ref{eq:Wxy}) gives
    \[
        W^{xy} = W^{01,10} = \id_{\Rreg} \otimes \SWAP_{\Areg}\otimes \SWAP_{\Breg}.
    \]
    Challenge $y=10$ tests $R_1B_1$ and
    $R_2A_2$, so $\cK_y$ retains $A_1$ and $B_2$; challenge $x=01$
    tests $R_1A_1$ and $R_2B_2$, so $\cK_x$ retains $A_2$ and $B_1$.
    
    We apply $M_{01,10}=C_{01}W^{xy}C_{10}^\dagger$ to a vector
    $\ket\psi := \ket u_{A_1}\otimes\ket v_{B_2}$ for $u,v\in\two$.  Applying the map $C_{10}^\dagger =\ket{\Phi^+}_{R_1B_1}\otimes\ket{\Phi^+}_{R_2A_2}\otimes\id_{A_1B_2}$ appends the Bell pairs of challenge $10$:
    \begin{align*}
        C_{10}^\dagger\ket\psi
        &= \left(\ket{\Phi^+}_{R_1B_1}\otimes\ket{\Phi^+}_{R_2A_2}\otimes\id_{A_1B_2}\right)(\ket u_{A_1}\otimes\ket v_{B_2}) \\
        &=\frac12\sum_{a,b\in\two}
        \ket a_{R_1}\otimes\ket b_{R_2}
        \otimes\ket u_{A_1}\otimes\ket b_{A_2}\otimes\ket a_{B_1}\otimes\ket v_{B_2};
    \end{align*}
    the swaps of $W^{xy}$ then carry $\ket u_{A_1}\otimes\ket b_{A_2}$ to
    $\ket b_{A_1}\otimes\ket u_{A_2}$ and
    $\ket a_{B_1}\otimes\ket v_{B_2}$ to
    $\ket v_{B_1}\otimes\ket a_{B_2}$:
    \[
        W^{xy}C_{10}^\dagger\ket\psi = \frac12\sum_{a,b\in\two}
        \ket a_{R_1}\otimes\ket b_{R_2}
        \otimes\ket b_{A_1}\otimes\ket u_{A_2}\otimes\ket v_{B_1}\otimes\ket a_{B_2};
    \]
    and the contractions
    $\bra{\Phi^+}_{R_1A_1}$ and $\bra{\Phi^+}_{R_2B_2}$ of $C_{01}$
    contribute a factor $\frac1{\sqrt2}\,\delta_{a,b}$ each.  Only the terms
    with $a=b$ survive, and there are two of them, so
    \[
        M_{01,10}\bigl(\ket u_{A_1}\otimes\ket v_{B_2}\bigr)
        =2\cdot\frac12\cdot\frac1{\sqrt2}\cdot\frac1{\sqrt2}\,
        \ket u_{A_2}\otimes\ket v_{B_1}
        =\frac12\,\ket u_{A_2}\otimes\ket v_{B_1}.
    \]
    Thus $M_{01,10}$ is $\tfrac12$ times the canonical unitary identification
    of $\mathcal K_{10}$ with $\mathcal K_{01}$, and $\norm{M_{01,10}}=1/2$.
\end{proof}

\begin{definition}[Cap matrix]\label{def:cap}
    For $x,y\in\{0,1\}^n$, define the $2^n\times2^n$ matrix $N_n$ by
    \[
        (N_n)_{x,y}=2^{-\max(t_A,t_B)}.
    \]
    On the diagonal, $t_A=t_B=0$, so $(N_n)_{x,x}=1$; for a chain, the entry equals $2^{-(t_A+t_B)}=2^{-t}$.

\end{definition}

Equation~\eqref{eq:gram}, Lemmas~\ref{lem:norm-matrix} and \ref{lem:chain},
and $M_{x,x}=\id_{\cK_x}$ (Section~\ref{sec:model}) imply
\begin{equation}\label{eq:cap-bound}
    \omega^*(\QCG_2^{\times n})\le2^{-n}\norm{N_n}.
\end{equation}

For $n=1$ the bound is tight: the matrix $N_1$ has spectrum $\{1/2,3/2\}$, so $2^{-1}\norm{N_1}=3/4$,
the exact one-copy value~\cite[Theorem~6]{CES25}.

\subsection{The two-copy cap matrix}\label{subsec:N2}

Order the challenges as $00,01,10,11$.  The diagonal entries are $M_{x,x}=\id$,
so $(N_2)_{x,x}=1$.  The four pairs at Hamming distance one,
\[
    \{00,01\},\quad \{00,10\},\quad \{01,11\},\quad \{10,11\},
\]
have cap $1/2$; the chain pair $\{00,11\}$ has cap $1/4$; and the loop pair
$\{01,10\}$ has cap $1/2$.  Hence
\begin{equation}\label{eq:N2-cap-matrix}
    N_2=
    \begin{pmatrix}
        1        & \tfrac12 & \tfrac12 & \tfrac14 \\[4pt]
        \tfrac12 & 1        & \tfrac12 & \tfrac12 \\[4pt]
        \tfrac12 & \tfrac12 & 1        & \tfrac12 \\[4pt]
        \tfrac14 & \tfrac12 & \tfrac12 & 1
    \end{pmatrix}.
\end{equation}

For the chain entry $M_{11,00}$, Lemma~\ref{lem:chain} gives
$\norm{M_{11,00}}\le2^{-t_B}=\tfrac14$; we now check that this cap holds
with equality for every choice of response unitaries.  Set
$W_A=U^{11}(U^{00})^\dagger$, acting on $A_1A_2E_A$, and
$W_B=V^{11}(V^{00})^\dagger$, acting on $B_1B_2E_B$, so that
$M_{11,00}=C_{11}(\id_{R_1R_2}\otimes W_A\otimes W_B)C_{00}^\dagger$ by Eq.~\eqref{eq:Wxy}.
The unitary $W_A$ is supported on $A_1A_2E_A$, disjoint from the registers
that $C_{11}$ contracts ($R_1B_1R_2B_2$), and $W_B$ is supported on
$B_1B_2E_B$, disjoint from the registers on which $C_{00}^\dagger$ inserts
Bell pairs ($R_1A_1R_2A_2$), so both commute past the contractions:
\[
    M_{11,00}
    =(W_A\otimes\id_{E_B})\,C_{11}C_{00}^\dagger\,(W_B\otimes\id_{E_A}),
\]
and it remains to compute $C_{11}C_{00}^\dagger$.  This product factors
over the two coordinates.  For each coordinate $j$, the corresponding
factor is $\tfrac12\Lambda_j^\dagger$, where
$\Lambda_j\colon A_j\to B_j$ is the coordinate-$j$ copy of the
identification $\Lambda$ from the single-coordinate relation
\eqref{eq:Lambda}. The adjoint appears because here the Bell bra
acts on $R_jB_j$ and the Bell ket lies on $R_jA_j$, reversing the
orientation.  Hence
\[
    M_{11,00}=\frac14\,(W_A\otimes\id_{E_B})
    \bigl(\Lambda_1^\dagger\otimes\Lambda_2^\dagger\otimes\id_{E_AE_B}\bigr)
    (W_B\otimes\id_{E_A})
\]
is $\tfrac14$ times a product of three unitaries, and
\begin{equation}\label{eq:M1100}
    \norm{M_{11,00}}=\frac14
    \qquad\text{for every choice of }U^{00},U^{11},V^{00},V^{11}.
\end{equation}

\subsection{The cap matrix for \texorpdfstring{$n$}{n} copies}\label{subsec:Nn}

Exchanging $x$ and $y$ exchanges $t_A$ and $t_B$, so $N_n$ is symmetric.
The extreme rows satisfy
\[
    (N_n)_{0^n,y}=2^{-|y|},
    \qquad
    (N_n)_{1^n,y}=2^{-(n-|y|)},
\]
where $|y|$ is the Hamming weight of $y$.
For example, at $n=3$ the entry $(N_3)_{000,111}=1/8$ is a chain, while
$(N_3)_{001,010}=1/2$ is a loop.

\section{Two-copy counterexample and challenge-independent strategies}
\label{sec:attack}

\begin{theorem}[Failure of strong parallel repetition]
    \label{thm:counterexample}
    The two-copy value satisfies
    \[
        \omega^*(\QCG_2^{\times2})
        \ge\frac{5+\sqrt{17}}{16}
        >\frac9{16}.
    \]
    Strong parallel repetition therefore fails for the unrestricted game.
\end{theorem}

\begin{remark}[Computer verification]\label{rem:code}
    The Python script \texttt{verify\_counterexample.py}, submitted as an
    ancillary file with this preprint, reproduces the computations in the
    proof below.
\end{remark}

\begin{proof}
    Take the private registers $E_A$ and $E_B$ to be one-dimensional.  Let $Y$
    be the Pauli $Y$ matrix and define
    \[
        G=(-iY)\otimes\ketbra{0}{0}+\id\otimes\ketbra{1}{1},
        \qquad
        -iY=\begin{pmatrix}0&-1\\1&0\end{pmatrix}.
    \]
    The gate $G$ applies $-iY$ to the first qubit when the second qubit is
    $\ket0$.
    Let $\mathrm{SW}$ interchange the qubits and set
    \begin{align*}
        U^{00}_{\Areg}=U^{11}_{\Areg}
        =V^{00}_{\Breg}=V^{11}_{\Breg} & =\id,                     \\
        U^{01}_{\Areg}=V^{10}_{\Breg}  & =G,                       \\
        U^{10}_{\Areg}=V^{01}_{\Breg}  & =\SWAP\,G\,\SWAP.
    \end{align*}

    For these responses, define
    \[
        H=\sum_{x\in\{0,1\}^2}
        S_x^\dagger\Pi_xS_x.
    \]
    Choosing $\rho$ as the rank-one projector onto a top eigenvector of $H$
    gives $\lambda_{\max}(H)/4$ for the value optimized in Eq.~\eqref{eq:value}.  In the computational basis, the
    two-copy acceptance projectors $\Pi_x$ have entries in
    $\tfrac14\mathbb Z$, and the response unitaries have integer entries, so
    $K=4H$ is a $64\times64$ integer matrix with characteristic polynomial
    \[
        \det(\lambda\id-K)=\lambda^{48}(\lambda-2)^2(\lambda-3)^3
        (\lambda-4)^2(\lambda-5)^3
        (\lambda^2-7\lambda+4)
        (\lambda^2-10\lambda+8)(\lambda^2-11\lambda+22).
    \]
    The largest root is $5+\sqrt{17}$; hence
    \[
        \frac14\lambda_{\max}(H)
        =\frac1{16}\lambda_{\max}(K)
        =\frac{5+\sqrt{17}}{16}.
    \]
    This exceeds $9/16$ since $\sqrt{17}>4$.

    Explicitly, in the computational basis with the qubits ordered
    $R_1,R_2,A_1,A_2,B_1,B_2$, the vector
    \begin{align*}
        w={} & \bigl(4\sqrt{17}-16\bigr)
        \bigl(\ket{000011}+\ket{001100}\bigr)
        +\bigl(13-3\sqrt{17}\bigr)
        \bigl(\ket{000110}+\ket{001001}\bigr)                           \\
             & +\bigl(4-\sqrt{17}\bigr)
        \bigl(\ket{010010}+\ket{011000}+\ket{100001}+\ket{100100}\bigr) \\
             & +\ket{010111}+\ket{011101}+\ket{101011}+\ket{101110}
        +\bigl(2\sqrt{17}-8\bigr)\ket{110000}+2\ket{111111}
    \end{align*}
    satisfies $Kw=(5+\sqrt{17})\,w$ by direct computation, with
    $\norm w^2=1972-476\sqrt{17}$, so choosing $\rho$ as the rank-one
    projector onto $w/\norm w$ attains the value $(5+\sqrt{17})/16$ for this strategy.
\end{proof}
\begin{remark}[Other gate choices]\label{rem:gate-family}
    The choice $-iY$ is not unique. Replacing it by any anti-diagonal
    unitary
    \[
        W_{\alpha,\beta}
        =
        \begin{pmatrix}
            0 & e^{i\alpha}\\
            e^{i\beta} & 0
        \end{pmatrix},
        \qquad \alpha,\beta\in\mathbb R,
    \]
    leaves the strategy value unchanged, since the corresponding
    averaged acceptance operators are unitarily equivalent under
    diagonal phase conjugations. This family includes $X$ and $Y$;
    the choice $-iY$ is convenient because $K$ has integer entries.
\end{remark}

\begin{definition}[Challenge-independent strategies]
    A strategy is \emph{challenge-independent} if $U^x_{\Areg E_A}$ and $V^x_{\Breg E_B}$ are independent of the challenge $x$.  In this case, we write $U_{\Areg E_A}$ and $V_{\Breg E_B}$ for $U^x_{\Areg E_A}$ and $V^x_{\Breg E_B}$, respectively.
\end{definition}

\begin{proposition}
    \label{prop:challenge-independent}
For every $n \ge 1$, the optimal value of $\QCG_2^{\times n}$ over
    challenge-independent strategies is exactly $(3/4)^n$.  In particular,
    every strategy with value exceeding $(3/4)^n$, such as that of
    Theorem~\ref{thm:counterexample} at $n=2$, is challenge-dependent.
\end{proposition}

\begin{proof}
    Write
    \[
        S=\id_{\Rreg}\otimes U_{\Areg E_A}\otimes V_{\Breg E_B}.
    \]
    Because the initial state is unrestricted and
    $\rho\mapsto S\rho S^\dagger$ is a bijection of the density operators,
    the fixed response unitaries may be absorbed into $\rho$.  Therefore the
    optimal challenge-independent value is
    \[
        \sup_\rho\frac1{2^n}\sum_{x\in\{0,1\}^n}
        \Tr\bigl[\Pi_xS\rho S^\dagger\bigr]
        =\sup_\rho\frac1{2^n}\sum_{x\in\{0,1\}^n}
        \Tr\bigl[\Pi_x\rho\bigr]=\biggl\lVert\frac1{2^n}\sum_{x\in\{0,1\}^n}\Pi_x\biggr\rVert.
    \]
    The operator inside the norm is positive semidefinite, so this value is
    attained by a top eigenvector.

    Grouping the registers coordinate by coordinate gives
    \[
        \frac1{2^n}\sum_{x\in\{0,1\}^n}\Pi_x
        =\left(\frac12\bigl(\Gamma^A+\Gamma^B\bigr)\right)^{\otimes n}
        \otimes\id_{E_AE_B},
    \]
    where $\Gamma^A$ and $\Gamma^B$ are the one-copy acceptance projectors
    from the single-coordinate calculation preceding
    Lemma~\ref{lem:chain} (Eq.~\eqref{eq:GammaAB}).  That calculation gives
    $\norm{\Gamma^A\Gamma^B}=\tfrac12$ (Eq.~\eqref{eq:GammaAB-norm}).

    By Halmos's two-subspaces theorem~\cite{Halmos69}, the projections
    $\Gamma^A$ and $\Gamma^B$ are jointly block-diagonalizable into blocks of
    dimension at most two, and on a two-dimensional block
    with singular value $s$ of $\Gamma^A\Gamma^B$ the sum $\Gamma^A+\Gamma^B$
    has eigenvalues $1\pm s$.  Maximizing over blocks gives
    \[
        \norm{\Gamma^A+\Gamma^B}=1+\norm{\Gamma^A\Gamma^B}=\frac32.
    \]
    Hence $\norm{\tfrac12(\Gamma^A+\Gamma^B)}=\tfrac34$, and since the
    operator norm is multiplicative under tensor products (with
    $\norm{\id_{E_AE_B}}=1$),
    \[
        \biggl\lVert\frac1{2^n}\sum_{x\in\{0,1\}^n}\Pi_x\biggr\rVert
        =\left(\frac34\right)^n.
    \]
\end{proof}

In the single-copy game, a challenge-independent strategy already achieves
the optimal value $3/4$~\cite[Theorem~6]{CES25}.  The advantage of
challenge-dependent responses therefore emerges only when the game is
repeated.

\section{Upper bounds from the cap matrix}\label{sec:upper}

\begin{corollary}[Two-copy bound]\label{cor:two-copy}
    \[
        \omega^*(\QCG_2^{\times2})
        \le\frac{11+\sqrt{65}}{32}.
    \]
\end{corollary}

\begin{proof}
    The spectrum of the two-copy cap matrix $N_2$ (Eq.~(\ref{eq:N2-cap-matrix}) in Section~\ref{subsec:N2}) is
    \[
        \left\{\frac12,\frac34,
        \frac{11-\sqrt{65}}8,
        \frac{11+\sqrt{65}}8\right\}.
    \]
    Equation~\eqref{eq:cap-bound} gives the result.
\end{proof}

Theorem~\ref{thm:counterexample} and Corollary~\ref{cor:two-copy} give
\eqref{eq:main-interval}, with both endpoints strict relative to
\eqref{eq:previous}.  In the permutation aggregation of
\cite[Theorem~8]{CES25}, the maximal directional caps for the four XOR
displacements $00,01,10,11$ are $1,\tfrac12,\tfrac12,\tfrac12$.  Retaining them
gives
\[
    \frac14\left(1+\frac12+\frac12+\frac12\right)=\frac58,
\]
whereas the block Gram comparison gives $(11+\sqrt{65})/32<5/8$.

\begin{proposition}[General upper bound]\label{prop:general}
    For every $n\ge1$,
    \[
        \omega^*(\QCG_2^{\times n})
        \le2^{-n}\norm{N_n}
        <\left(\frac12+\frac1{2\sqrt2}\right)^n.
    \]
\end{proposition}

\begin{proof}
    Let $\widehat N_n$ be the matrix with entries
    \[
        (\widehat N_n)_{x,y}=2^{-|x\oplus y|/2},
    \]
    where $|x\oplus y|$ is the Hamming distance between $x$ and $y$.
    Since $\max(t_A,t_B)\ge(t_A+t_B)/2 = |x\oplus y|/2$ (notation of
    Section~\ref{sec:caps}), Definition~\ref{def:cap} gives
    $0\le N_n\le \widehat N_n$
    entrywise, with strict inequality on every distinct chain pair.

    The characters $\chi_s(y)=(-1)^{s\cdot y}$ diagonalize $\widehat N_n$, where
    $s\cdot y=\sum_js_jy_j\pmod 2$ and $|s|$ is the Hamming weight of $s$.
    Since the entries depend only on
    $x\oplus y$, substituting $z=x\oplus y$ gives
    \begin{align*}
        (\widehat N_n\chi_s)(x)
         & =\sum_{y\in\{0,1\}^n}2^{-|x\oplus y|/2}(-1)^{s\cdot y} \\
         & =\chi_s(x)\prod_{j=1}^n
        \left(1+(-1)^{s_j}2^{-1/2}\right)                         \\
         & =\chi_s(x)
        \left(1+2^{-1/2}\right)^{n-|s|}
        \left(1-2^{-1/2}\right)^{|s|}.
    \end{align*}
    Thus $\chi_s$ is an eigenvector with eigenvalue
    \[
        \left(1+2^{-1/2}\right)^{n-|s|}
        \left(1-2^{-1/2}\right)^{|s|}.
    \]
    These eigenvalues are positive and are maximized at $s=0$. Hence
    \[
        \norm{\widehat N_n}=(1+2^{-1/2})^n
        =2^n\left(\frac12+\frac1{2\sqrt2}\right)^n.
    \]
    Since $N_n$ is symmetric with strictly positive entries, Perron--Frobenius
    gives a positive unit Perron eigenvector $v$, so
    $\norm{N_n}=v^{\mathsf T}N_nv$. The matrix $\widehat N_n-N_n$ is
    entrywise nonnegative and has a positive entry at, for example, the
    chain pair $0^n,10^{n-1}$. Therefore,
    \[
        \norm{N_n}
        =v^{\mathsf T}N_nv
        <v^{\mathsf T}\widehat N_nv
        \leq \norm{\widehat N_n}.
    \]
    Equation~\eqref{eq:cap-bound} proves the claim.
\end{proof}

\section{Conclusion}

We have shown that strong parallel repetition fails for the unrestricted
two-player cloning game: an explicit challenge-dependent two-copy strategy
achieves $(5+\sqrt{17})/16>9/16$, whereas by
Proposition~\ref{prop:challenge-independent} challenge-independent
strategies have value exactly $(3/4)^n$ for every $n$, so any advantage
over the strong parallel-repetition value $(3/4)^n$ requires
challenge-dependent responses. On the upper side, retaining the directional
overlap information from the proof of \cite[Theorem~8]{CES25} in a block
Gram matrix gives the bound $2^{-n}\norm{N_n}$, strictly below the previous
bound for every $n$.  The exact two-copy value remains open in the interval
\eqref{eq:main-interval}.

This lower-bound strategy concerns the unrestricted cloning game.  It does not
by itself yield an attack on the routing quantum position-verification protocol
in the no-pre-shared-entanglement model of \cite[Section~5]{CES25}, because that
model constrains the initial state and the referee's marginal.  Instantaneous
nonlocal quantum computation provides a generic entanglement-assisted attack on
position-verification schemes~\cite{BeigiKonig11}.  For recent progress on
parallel repetition in the adjacent routing setting, see the analysis of the
$f$-BB84 and $f$-routing protocols by Escol\`a-Farr\`as and
Speelman~\cite{EFS25}.

\paragraph{Acknowledgments.}
This work was carried out through the REU in Quantum Information at the
University of South Carolina.  This material is based upon work supported
by the U.S.\ National Science Foundation under Grant No.~2447817.

\end{document}